\documentclass[11pt]{article}
\usepackage[utf8]{inputenc}

\usepackage{graphicx}
\usepackage{xcolor}
\usepackage{amsmath,amsthm,amssymb}
\usepackage[margin=1in]{geometry}
\usepackage[pdftex,colorlinks=true,linkcolor=blue,citecolor=blue,urlcolor=blue]{hyperref}
\usepackage{comment}
\usepackage[normalem]{ulem}
\usepackage{cite}
\usepackage{tcolorbox}

\newcommand{\E}{\mathbb{E}}

\def\ba{\begin{array}}
\def\ea{\end{array}}

\def\0{{\bf 0}}

\def\b{{\bf b}}

\def\x{{\bf x}}

\def\u{{\bf u}}
\def\E{{\mathbb E}}
\def\D{{\mathcal D}}

\def\F{{\mathcal F}}
\def\T{{\mathcal T}}
\def\Tr{{\rm Tr}}

\newcommand{\ket}[1]{| #1 \rangle}
\newcommand{\bra}[1]{\langle #1|}

\renewcommand{\Re}{\operatorname{Re}}

\newcommand{\be}{\begin{equation}}
\newcommand{\ee}{\end{equation}}
\newcommand{\bea}{\begin{eqnarray}}
\newcommand{\eea}{\end{eqnarray}}
\newcommand{\bes}{\begin{equation*}}
\newcommand{\ees}{\end{equation*}}
\newcommand{\beas}{\begin{eqnarray*}}
\newcommand{\eeas}{\end{eqnarray*}}

\makeatletter
\newtheorem*{rep@theorem}{\rep@title}
\newcommand{\newreptheorem}[2]{%
\newenvironment{rep#1}[1]{%
 \def\rep@title{#2 \ref{##1} (restated)}%
 \begin{rep@theorem}}%
 {\end{rep@theorem}}}
\makeatother

\allowdisplaybreaks

\newtheorem{thm}{Theorem}
\newtheorem*{thm*}{Theorem}

\newtheorem{lem}[thm]{Lemma}
\newtheorem*{lem*}{Lemma}

\newtheorem{prop}[thm]{Proposition}

\newtheorem{fact}[thm]{Fact}

\newreptheorem{thm}{Theorem}
\newreptheorem{lem}{Lemma}
\newreptheorem{cor}{Corollary}

\usepackage{pgfplots}

\usepackage{authblk}

\title{Worst-case Harrow--Hassidim--Lloyd algorithm with average-case-correct quantum Fourier transform}

\author{Changpeng Shao\thanks{changpeng.shao@amss.ac.cn}}

\affil{\small{SKLMS, Academy of Mathematics and Systems Science, Chinese Academy of Sciences, Beijing, China}}

\date{\today}

\begin{document}

\maketitle

\begin{abstract}

In [\href{https://quantum-journal.org/papers/q-2022-12-07-872/}{Quantum 6, 872, 2022}], Linden and de Wolf proposed a lightweight protocol for verifying average-case correctness of the quantum Fourier transform (QFT). They showed that good average-case QFT performance is sufficient for good worst-case performance in several quantum information-processing tasks. In this work, we study whether such average-case guarantees are also sufficient when the QFT is used coherently inside the Harrow--Hassidim--Lloyd algorithm. We show that the original average-case condition is not quite strong enough for this purpose, due to possible relative phase errors between different eigenspaces. To address this, we introduce a strengthened Linden--de Wolf-type verification condition that controls the relevant coherences, and prove that it guarantees worst-case correctness of the HHL algorithm in several natural settings.

\end{abstract}

\section{Introduction}

Quantum verification refers to the process of certifying the correctness and reliability of quantum devices and quantum computations. Given the probabilistic nature of quantum systems, verification is a central issue in the practical use of quantum computers. Much effort has been devoted to quantum verification in different settings; see, for example, \cite{linden2021lightweight,eisert2020quantum,mahadev2018classical,flammia2011direct,asadi2022quantum,gheorghiu2019verification,da2011practical,TothGuhne2005PRL,TothGuhne2005PRA,TothEtAl2010PRL}.

The quantum Fourier transform (QFT) is a basic primitive in quantum computing, with applications including Shor's algorithm~\cite{shor1997polynomial} and quantum phase estimation~\cite{kitaev1995quantum}. Recently, Linden and de Wolf~\cite{linden2021average} proposed an efficient lightweight protocol for verifying average-case correctness of the QFT. They showed that good average-case QFT performance is sufficient for good worst-case performance in several key quantum tasks, including phase estimation, period finding, and amplitude estimation. This is useful because worst-case verification can be exponentially hard, whereas average-case verification can be performed efficiently.
Very recently, Kahanamoku-Meyer, Blue, Bergamaschi, Gidney, and Chuang introduced the notion of optimistic quantum circuits and constructed a logarithmic-depth in-place QFT that is correct on most inputs~\cite{kahanamoku2025log}. Their work is complementary to the average-case verification viewpoint of Linden and de Wolf: it provides explicit efficient QFT circuits with good average-case behavior and studies when such circuits can be used in worst-case settings.

In this work, we study whether an average-case-correct QFT can also be used coherently inside the Harrow--Hassidim--Lloyd (HHL) algorithm~\cite{harrow2009quantum} for solving linear systems. The HHL algorithm is an important application of phase estimation and has played a central role in the development of quantum algorithms, especially in quantum machine learning; see, for example, \cite{biamonte2017quantum,lloyd2014quantum,rebentrost2014quantum}. At first sight, the result of Linden and de Wolf for phase estimation might seem sufficient for HHL. However, their result concerns the task of outputting eigenvalue estimates. In contrast, HHL must output a quantum state whose amplitudes depend coherently on the eigenvalues and eigenvectors of the input matrix. For this purpose, average-case correctness in the sense of Linden and de Wolf is not enough: an implementation of the inverse QFT may pass their test while introducing uncontrolled relative phases between different eigenvalue components.

To illustrate this point, let $C=DF_N^{-1}$, where $F_N^{-1}$ is the inverse QFT and $D$ is an unknown diagonal unitary. Then $C$ is perfectly correct on the Fourier basis in the sense that  $\E_k[\F(C \ket{\hat{k}} , F_N^{-1} \ket{\hat{k}} ) ] \geq 1-\eta$,
where $\ket{\hat{k}}=F_N\ket{k}$ and  $\F$ is the fidelity. Thus $C$ passes the Linden--de Wolf test with probability one. However, if the input state is
$\alpha_1\ket{\widehat{k_1}}\ket{u_1}+\alpha_2\ket{\widehat{k_2}}\ket{u_2}$,
then applying $C\otimes I$ produces
$\alpha_1 e^{i\theta_1}\ket{k_1}\ket{u_1}
+\alpha_2 e^{i\theta_2}\ket{k_2}\ket{u_2}$
for some unknown phases $\theta_1,\theta_2$. If the goal is only to measure the first register and output $k_1$ or $k_2$, these phases are irrelevant. But if the goal is to output the coherent state
$\alpha_1\ket{k_1}\ket{u_1}+\alpha_2\ket{k_2}\ket{u_2}$
up to a prescribed error, then this guarantee is not sufficient unless $\theta_1\approx\theta_2$. This phase ambiguity is problematic for the HHL algorithm.

To resolve this issue, we introduce a strengthened Linden--de Wolf-type condition that controls the relevant coherences. Let $C$ be a quantum channel. Linden and de Wolf assumed that
\be
\label{intro:eq1}
\E_k \left[ \F(C(\ket{\hat{k}}\bra{\hat{k}}), \ket{k}\bra{k} ) \right]
=
\E_{k} \left[ \bra{k} \, C(\ket{\hat{k}}\bra{\hat{k}}) \, \ket{k} \right]
\geq 1- \eta,
\ee
that is, $C$ is close to $F_N^{-1}$ on the Fourier basis on average. In this work, we consider the stronger condition
\be
\label{intro:eq2}
\E_{k,l} \left[\bra{k} \, C(\ket{\hat{k}}\bra{\hat{l}}) \, \ket{l} \right] \geq 1- \eta.
\ee
In the special case where $C$ is unitary, this condition implies that $|\E_k \bra{k} C\ket{\hat{k}}|^2 \geq 1-\eta,$
or equivalently, the normalized trace of $CF_N$ is close to one. Thus $CF_N$ is close to the identity up to a global phase.

We also show that this strengthened condition can still be verified in a lightweight way. More precisely, $C$ satisfies the above coherence condition if and only if it is close to $F_N^{-1}$ on average both on the Fourier basis and on the computational basis. Combining this with the protocol of Linden and de Wolf gives a lightweight test for the strengthened condition. Similar results hold for quantum channels $P$ that approximate $F_N$ on average, under the analogous condition
\be
\label{intro:eq3}
\E_{k,l} \left[\bra{\widehat{-k}} \, P(\ket{k}\bra{l}) \, \ket{\widehat{-l}} \right] \geq 1- \eta.
\ee

Finally, assuming that $C$ and $P$ satisfy these strengthened average-case conditions, we prove that replacing $F_N^{-1}$ by $C$ and $F_N$ by $P$ in the HHL algorithm produces a state close to the desired output state, with high average fidelity. The rigorous statements are given in Theorems~\ref{main theorem} and~\ref{main theorem:general case}. We also consider special cases in which $C$ and $P$ are unitary channels; in these settings the results take a simpler form, as shown in Theorems~\ref{thm1:case with C inverse}, \ref{thm2:case with C inverse}, \ref{thm1:case with C, P and more}, and \ref{thm2:case with C, P and more}.

\section{Preliminaries}
\label{section:Preliminaries}

\subsection{Some facts about fidelity and trace distance}

The {\em fidelity} between two mixed quantum states $\rho$ and $\sigma$ is defined as
\be
\F(\rho,\sigma) := \left( {\rm Tr}\sqrt{\sqrt{\sigma} \rho \sqrt{\sigma}} \, \right)^2.
\ee
For any two states $\rho$ and $\sigma$, we have $0\leq \F(\rho,\sigma) \leq 1$ and $\F(\rho,\sigma) = 1$ if and only if $\rho=\sigma$.
In particular, if $\sigma=\ket{\psi} \bra{\psi}$ is a pure state, then $\F(\rho,\ket{\psi} \bra{\psi})
=\langle \psi| \rho | \psi\rangle$. Let $\mathcal{E}$ be a quantum channel, then $\F(\rho, \sigma) \leq \F(\mathcal{E}(\rho), \mathcal{E}(\sigma))$. The equality holds when $\mathcal{E}$ is a unitary channel. 

A closely related notion is {\em trace distance}.
It is defined as follows
\be
\T(\rho, \sigma) := \frac{1}{2} \|\rho - \sigma\|_1 = \frac{1}{2} {\rm Tr} \sqrt{(\rho-\sigma)^\dag (\rho-\sigma)}.
\ee
The trace distance is a metric on the space of density matrices, i.e., it is non-negative, symmetric, and satisfies the triangle inequality. 
Moreover,
let $\mathcal{E}$ be a quantum channel, then $ \T(\mathcal{E}(\rho), \mathcal{E}(\sigma))  \leq \T(\rho, \sigma)$, and equality holds when $\mathcal{E}$ is a unitary channel. 
The trace distance can be computed using the singular values of $\rho-\sigma$. More precisely, it equals half of the sum of all singular values of $\rho-\sigma$.

The fidelity is related to the trace distance by the following inequalities:
\bea
&& 1 - \sqrt{\F(\rho,\sigma)} \leq \T(\rho, \sigma) \leq \sqrt{1-\F(\rho,\sigma)} , 
\label{ineq1} \\
&& 1 - \F(\rho, \ket{\psi} \bra{\psi} ) \leq \T(\rho, \ket{\psi} \bra{\psi}).
\label{ineq2} 
\eea
Fidelity does not satisfy the triangle inequality; however, according to (\ref{ineq1}), (\ref{ineq2}), we have the following inequality that will be used frequently in this work.
For more about properties of fidelity and trace distance, we refer to the book \cite{watrous2018theory}.

\begin{prop}
\label{prop:inequality for fidelity}
For any two density operators $\rho,\sigma$ and any quantum state $\ket{\psi}$, we have
\bes
\F(\rho,\ket{\psi}\bra{\psi}) \geq
1 - \sqrt{1-\F(\rho, \sigma)} - \sqrt{1-\F(\sigma, \ket{\psi} \bra{\psi})}.
\ees
\end{prop}

\begin{proof}
From (\ref{ineq2}), we have $\F(\rho,\ket{\psi}\bra{\psi}) \geq 1 - \T(\rho, \ket{\psi} \bra{\psi}).$ By the triangle inequality of trace distance, we have that $\F(\rho,\ket{\psi}\bra{\psi}) \geq
1 - \T(\rho, \sigma) - \T(\sigma, \ket{\psi} \bra{\psi}).$
Finally, by (\ref{ineq1}), we obtain the claimed lower bound.
\end{proof}

\begin{fact}[von Neumann's trace inequality \cite{mirsky1975trace}]
\label{fact:von Neumann's trace inequality}
For any $n \times n$ operators $A,B$ with singular values 
$\alpha_1 \geq \alpha_2 \geq \cdots \geq \alpha_n$ and 
$\beta_1 \geq \beta_2 \geq \cdots \geq \beta_n$
respectively, we have ${\rm Tr}(AB) \leq \sum_{i=1}^n \alpha_i \beta_i.$
Particularly, ${\rm Tr}(AB) \leq \|A\|_2 \|B\|_1$, where $\|\cdot\|_2$ is the operator norm and $\|\cdot\|_1$ is the trace norm.
\end{fact}

\begin{fact}[Jensen's inequality]
\label{fact:Jensen's inequality}
Let $X$ be a random variable and $\varphi$ be a convex function, then
$\varphi(\E[X]) \leq \E[\varphi(X)]$.
\end{fact}

\subsection{Sets of quantum channels that are close to the QFT on average}

For any integer $N$, we denote by $[N]=\{0,1,\ldots,N-1\}$ and $\omega_N=e^{2\pi i/N}$. The quantum Fourier transform (QFT) is defined as $F_N = N^{-1/2}\sum_{j,k\in[N]}\omega_N^{jk} \ket{j} \bra{k}$. In this work, we always assume that $N=2^n$ for some integer $n>0$.
For any integer $k\in[N]$, we use $\ket{\hat{k}}$ to denote the Fourier basis state, i.e.,
\[
\ket{\hat{k}} 
= \frac{1}{\sqrt{N}} \sum_{l\in [N]} \omega_N^{kl} \, \ket{l}
= \bigotimes_{l=1}^n \Big(\ket{0} + e^{2\pi i k/2^l} \ket{1} \Big).
\]
Let
\be
\label{unitary:Fourier basis}
U_k = \bigotimes_{l=1}^n 
\begin{pmatrix}
1 & 0 \\
0 & e^{2\pi i k/2^l} 
\end{pmatrix},
\ee
then $\ket{\hat{k}} = U_k \ket{+}^n$. It is easy to check that $U_k^{-1} = U_{-k}$.

In the following, we will define three different sets collecting all quantum channels that are close to the inverse of QFT in different senses. First, inspired by \cite{linden2021average},
for any $\eta\in[0,1]$, we define
\bea
S_1(\eta) &=& \Big\{C \text{ is a quantum channel}: 
\E_k\Big[\F\Big(C (\ket{\hat{k}}\bra{\hat{k}}) , \, \ket{k}\bra{k}\Big) \Big] \geq 1-\eta \Big\} \\
&=& \Big\{C \text{ is a quantum channel}: \E_k\Big[\bra{k} \, C\circ F_N (\ket{k}\bra{k}) \, \ket{k}\Big] \geq 1-\eta \Big\},
\eea
In the above, $C \circ F_{N} $ refers to the composition of two channels, i.e., $C \circ F_{N} (\rho) = C (F_N \rho F_N^\dag)$ for any density operator $\rho$.
The set $S_1(\eta)$ collects all channels $C$ that are supposed to be close to $F_{N}^{-1}$ on average on the Fourier basis. The following is a lightweight protocol for verifying this.

\begin{quote}
{\bf Testing Algorithm 1 (TA1).} \cite[Theorem 1]{linden2021average} {\em Choose $k\in[N]$ uniformly at random and prepare $\ket{\hat{k}}$. Run $C$ on $\ket{\hat{k}}$ and measure the resulting state in the computational basis. Output 1 if the outcome is $k$, and output 0 otherwise.}
\end{quote}
It was proved in \cite[Theorem 1]{linden2021average} that to test if $C\in S_1(\eta)$ or not it suffices to run the above protocol $O(\eta^{-2}\log(1/\delta))$ times, where $\delta$ is the failure probability. 
As explained in the introduction, $C\in S_1(\eta)$ is not strong enough for us to run the HHL algorithm.
To overcome this, we define a new set as follows:
\be
S_3(\eta) = \{C \text{ is a quantum channel}: \E_{k,l} \bra{k} \, C(\ket{\hat{k}}\bra{\hat{l}}) \, \ket{l} \geq 1- \eta\}.
\ee

To obtain some intuition about $S_3(\eta)$, we consider the case that $C \in S_3(\eta)$ is a unitary channel.
We can check that the norm of the normalized trace of $CF_N$ is close to 1, i.e., $|\E_k\langle k| C | \hat{k} \rangle|^2 \geq 1-\eta$. In this case, we will have $C \approx e^{i\theta} F_N^{-1}$ for some $\theta$, which only introduces a global phase up to certain error.
We shall prove that $S_3(\eta)$ is the right set for us to run the HHL algorithm. Before that, the first question we need to answer is how to verify whether $C \in S_3(\eta)$ in a lightweight way. For this, we introduce another set containing all $C$ that are close to $F_N^{-1}$ on average on the computational basis:
\be
S_2(\eta) = \Big\{C \text{ is a quantum channel}: 
\E_k\Big[\F\Big(C (\ket{k}\bra{k}) , \, \ket{\widehat{-k}}\bra{\widehat{-k}}\Big) \Big] \geq 1-\eta \Big\} .
\ee
In the above and also in the following, by $-k$ we always mean it is modulo $N$, i.e., $\ket{-k} = \ket{N-k}$. 

The following result states that to verify $C\in S_3(\eta)$ it suffices to propose a lightweight protocol to verify whether $C \in S_2(\eta)$. The proof of this theorem is given at the end of Section \ref{section:Some facts about S3}.

\begin{thm}
\label{thm:testing S3}
Suppose $C \in S_1(\eta_1) \cap S_2(\eta_2)$, then $C \in S_3(\eta_1+\eta_2)$.
Conversely, if $C\in S_3(\eta)$, then $C \in S_1(\eta) \cap S_2(\eta)$.
\end{thm}

By Theorem \ref{thm:testing S3}, if a testing algorithm indicates that $C \in S_1(\eta)\cap S_2(\eta)$, then we know that $C \in S_3(2\eta)$. If $C \not\in S_1(\eta)\cap S_2(\eta)$, we may not have
$C \not \in S_3(2\eta)$. But what we can claim is $C \not \in S_3(\eta)$. 
In practice, we can just focus on quantum channels in $S_1(\eta)\cap S_2(\eta)$.
When $\eta$ is small, this difference is negligible. So for us, it suffices to find a lightweight protocol to verify whether $C \in S_1(\eta)\cap S_2(\eta)$ or not.
With {\bf TA1}, we know how to test if $C \in S_1(\eta)$. Below, we propose a lightweight protocol for verifying if $C \in S_2(\eta)$ or not.

Recall that $F_N^{-1} \ket{k} = \ket{\widehat{-k}} = U_{-k}\ket{+}^n = U_{-k}H^{\otimes n}\ket{0}^n$. Based on this, we can propose the following protocol.

\begin{quote}
{\bf Testing Algorithm 2 (TA2).} {\em Choose $k$ uniformly at random. Run $H^{\otimes n} U_k C$ on $\ket{k}$ and measure the resulting state on the computational basis. Output 1 if the outcome is $\ket{0}^n$, and output 0 otherwise.}
\end{quote}

The intuition behind the above two testing algorithms: We shall suppose that $\eta = 0$ and $C$ is a unitary channel for simplicity.
If $C$ passes {\bf TA2}, then we have $C = F_N^{-1} \Phi$ where $\Phi$ is diagonal and unitary. Note that if $C$ also passes {\bf TA1}, then in the special case that $\eta=0$ we have $C = \Psi F^{-1} $ where $\Psi$ is diagonal and unitary. Thus, if $C$ passes both {\bf TA1} and {\bf TA2}, we will have $C = F_N^{-1}$ up to a global phase. 

The correctness of {\bf TA2} is similar to that of {\bf TA1}. Similar to \cite[Theorem 1]{linden2021average}, we can estimate $\E_k[\F(C (\ket{k}\bra{k}), \ket{\widehat{-k}} \bra{\widehat{-k}} )]$ up to additive error $\varepsilon$, with success probability $1-\delta$, by running the above protocol $O(\varepsilon^{-2}\log(1/\delta))$ times.

\section{Some facts about $S_3(\eta)$}
\label{section:Some facts about S3}

In this section, we present some basic properties of the set $S_3(\eta)$.
Let $C(\rho) = \sum_{i=1}^r A_i \rho A_i^\dag$ be a Kraus representation of $C$ satisfying $\sum_{i=1}^r A_i^\dag A_i=I$, then $C\in S_3(\eta)$ implies
\be
\label{0204eq1}
1-\eta \leq \sum_{i=1}^r \E_{k,l} \bra{k} A_i \ket{\hat{k}} \bra{\hat{l}} A_i^\dag \ket{l}
= \sum_{i=1}^r \big| \E_k[\bra{k} A_i \ket{\hat{k}}] \big|^2.
\ee
Particularly, when $r=1$, i.e., when $C(\rho)=U\rho U^\dag$ is a unitary channel for some unitary $U$, we have $|\E_k[\bra{k} U \ket{k}]|^2 \geq 1-\eta$. The above also implies that 
$\E_{k,l} \bra{k} \, C(\ket{\hat{k}}\bra{\hat{l}}) \, \ket{l}$ is always a nonnegative real number.
The description (\ref{0204eq1}) of $C\in S_3(\eta)$ will be useful in the analysis of implementing the HHL algorithm using noisy QFT.

As a quantum channel, an equivalent description is that there is a unitary $U$ such that $C(\rho) = {\rm Tr}_2 [U (\rho \otimes \ket{0^m}\bra{0^m})U^\dag]$ for some integer $m$. So if $C \in S_3(\eta)$, then we have
\bea
1-\eta &\leq& \E_{k,l}\bra{k} {\rm Tr}_2 [U (\ket{\hat{k}} \bra{\hat{l}}   \otimes \ket{0^m}\bra{0^m})U^\dag] \ket{l} \nonumber \\
&=& {\rm Tr}_2 \Big[\E_{k,l} (\bra{k}\otimes I)  U \ket{\hat{k}, 0^m} \bra{\hat{l}, 0^m}U^\dag (\ket{l} \otimes I) \Big]  \nonumber \\
&=& {\rm Tr}_2 \Big[  \Big(\E_{k} (\bra{k}\otimes I)  U \ket{ \hat{k}, 0^m} \Big) \Big(\E_{k} \bra{\hat{k}, 0^m}U^\dag (\ket{k} \otimes I) \Big) \Big]  \nonumber \\
&=& \sum_{p\in\{0,1\}^m} \big|\E_{k} \bra{k,p}  U \ket{ \hat{k}, 0^m} \big|^2 .
\label{description2}
\eea
The above description of $C\in S_3(\eta)$ will be useful in the proof of Theorem \ref{thm:testing S3}.

The following result states that if $C \in S_3(\eta)$, then the phases of 
$\{\bra{k} \, C (\ket{\hat{k}}\bra{\hat{l}}) \, \ket{l} : k,l\in [N]\}$ are close to each other on average. This is a desired property when running the HHL algorithm. In the following result, we consider a general result.
\begin{prop}
\label{prop:fact1 of S3}
Assume that $x_p = r_p e^{i\theta_p}$ for $p\in [M]$, where $r_p \in[0,1]$. If $|\E_p[ x_p]|^2 \geq 1 - \eta$, then $|\E_p[e^{i\theta_p}]|^2 =\E_{p,q}[\cos(\theta_p-\theta_q)] \geq 1 - 2\eta$.
\end{prop}

\begin{proof}
Denote
$A_{+}=\{(p,q): \cos(\theta_p - \theta_q)\geq 0\}, A_{-}=\{(p,q): \cos(\theta_p - \theta_q) < 0\}$. Then
\beas
|\E_p[x_p]|^2
&=& \frac{1}{M^2} \sum_{p,q=1}^M r_pr_q \cos(\theta_p - \theta_q) \\
&=& \frac{1}{M^2} \sum_{(p,q)\in A_+} r_pr_q \cos(\theta_p - \theta_q)
-
\frac{1}{M^2} \sum_{(p,q)\in A_-} r_pr_q |\cos(\theta_p - \theta_q)|.
\eeas
Since $|\E_p[x_p]|^2\geq 1-\eta$, the above equality implies $|A_-| \leq M^2 \eta$ and
\[
\sum_{(p,q)\in A_+} \cos(\theta_p - \theta_q) \geq M^2 (1-\eta).
\]
Therefore,
\beas
\sum_{(p,q)\in A_+} \cos(\theta_p - \theta_q)
-
\sum_{(p,q)\in A_-} |\cos(\theta_p - \theta_q)|
\geq M^2 (1-\eta) - M^2 \eta
=M^2 (1-2\eta).
\eeas
In conclusion, we have
\beas
\E_{p,q}[\cos(\theta_p-\theta_q)]
= \frac{1}{M^2} \sum_{(p,q)\in A_+} \cos(\theta_p - \theta_q)
-
\frac{1}{M^2} \sum_{(p,q)\in A_-} |\cos(\theta_p - \theta_q)|
\geq 1-2\eta.
\eeas
This completes the proof.
\end{proof}

The next result states that if $C\in S_3(\eta)$, then $C$ is correct on any orthogonal basis on average. Combining Theorem \ref{thm:testing S3}, this implies that if $C$ is close to $F_N^{-1}$  on average both on the Fourier basis and on the computational basis, then $C$ is close to $F_N^{-1}$ on any orthogonal basis on average.

\begin{prop}
\label{prop:fact2 of S3}
Let $\{\ket{\u_k}\}_{k\in[N]}$ be an orthogonal basis and $C\in S_3(\eta)$, then
$$
\E_{k,l} \bra{\u_k} \, C\circ F_N (\ket{\u_k}\bra{\u_l}) \, \ket{\u_l} \geq 1- \eta.
$$
In particular, we have
$S_3(\eta)=\{C \text{ is a quantum channel}: \E_{k,l} \bra{\widehat{-k}} \, C(\ket{k}\bra{l}) \, \ket{\widehat{-l}} \geq 1- \eta\}$.
\end{prop}

\begin{proof}
Assume that $\ket{\u_k} = \sum_l u_{kl} \ket{l}$, then it is easy to compute that
\beas
\E_{k,l} \bra{\u_k} \, C\circ F_N (\ket{\u_k}\bra{\u_l}) \, \ket{\u_l} 
&=& \sum_{p,q,r,s} \E_{k} [\bar{u}_{kp}u_{kq} ] \E_l [\bar{u}_{lr} u_{ls} ] \bra{p} \, C (\ket{\hat{q}}\bra{\hat{r}}) \, \ket{s} \\
&=& \frac{1}{N^2}\sum_{p,r} \bra{p} \, C (\ket{\hat{p}}\bra{\hat{r}}) \, \ket{r} .
\eeas
Since $C\in S_3(\eta)$, we have $\frac{1}{N^2}\sum_{p,r} \bra{p} \, C (\ket{\hat{p}}\bra{\hat{r}}) \, \ket{r} \geq 1-\eta$.
In the above, we used the fact that $\{\ket{\u_k}\}_{k\in[N]}$ is an orthogonal basis so that $\E_{k} [\bar{u}_{kp}u_{kq} ] = \delta^p_q$.
\end{proof}

Intuitively, $C\in S_3(\eta) \subseteq S_1(\eta)$ implies that most diagonal entries of $C(\ket{\hat{l}}\bra{\hat{l}})$ have absolute values close to 1, so most non-diagonal entries should be small. With this intuition, we have the following result.

\begin{prop}
\label{prop:nondiagonal entries of S3}
Assume that $C\in S_1(\eta)$, then for any $k\neq 0$, we have $\E_l [\bra{k+l} \, C(\ket{\hat{l}} \bra{\hat{l}}) \, \ket{k+l}] \leq \eta$.
\end{prop}

\begin{proof}
As a quantum channel, we know that $C(\ket{\hat{l}} \bra{\hat{l}})$ is a density operator. So all diagonal entries are nonnegative real numbers and their sum is 1. Thus we have
\[
\bra{k+l} \, C(\ket{\hat{l}} \bra{\hat{l}}) \, \ket{k+l}
\leq 
1 - \bra{l} \, C(\ket{\hat{l}} \bra{\hat{l}}) \, \ket{l}.
\]
Since $C\in S_1(\eta)$, we have 
$\E_l[\bra{l} \, C(\ket{\hat{l}} \bra{\hat{l}}) \, \ket{l}]\geq 1-\eta$. 
This means
$\E_l[\bra{k+l} \, C(\ket{\hat{l}} \bra{\hat{l}}) \, \ket{k+l}] \leq \eta$, which is as claimed.
\end{proof}

In the end, we aim to prove Theorem \ref{thm:testing S3}. 

\begin{proof}[Proof of Theorem \ref{thm:testing S3}]
We first prove the first claim, i.e., if $C \in S_1(\eta_1) \cap S_2(\eta_2)$, then $C \in S_3(\eta_1+\eta_2)$.
We will use (\ref{description2}).
As a quantum channel, we assume that there is a unitary $U$ such that
$C(\rho) = \Tr_2(U (\rho\otimes \ket{0..0}\bra{0..0}) U^\dag)$.
Assume that $U \ket{\hat{k},0..0} = \sum_{l,m} U^k_{l,m} \ket{l,m}$.
By assumption, we have $C\in S_1(\eta_1) \cap S_2(\eta_2)$, so
\be
\label{ineqs}
\frac{1}{N} \sum_{q,k} |U^k_{k,q}|^2 \geq 1-\eta_1, \quad
\frac{1}{N^3} \sum_{q,k} \left|\sum_{l,m} \omega^{k(m-l)} U^l_{m,q} \right|^2 \geq 1-\eta_2.
\ee
For the second inequality, we can expand it as follows
\beas
1-\eta_2 &\leq& \frac{1}{N^3} \sum_{q,k} \left|\sum_{l} U^l_{l,q} + \sum_{l\neq m} \omega^{k(m-l)} U^l_{m,q}  \right|^2 \\
&=& \frac{1}{N^3} \sum_{q,k} \left[ \left|\sum_{l} U^l_{l,q}\right|^2 + \left|\sum_{l\neq m} \omega^{k(m-l)} U^l_{m,q}  \right|^2
+2 \Re\left(\sum_{l} \overline{U}^l_{l,q}\right) \left(\sum_{l\neq m} \omega^{k(m-l)} U^l_{m,q}  \right) \right].
\eeas
The summation over $k$ implies that the third term is 0. The first term is
\[
\frac{1}{N^2} \sum_{q} \left|\sum_{l} U^l_{l,q}\right|^2 
= \E_{k,l} \bra{k} \, C(\ket{\hat{k}}\bra{\hat{l}}) \, \ket{l}.
\]
By the Cauchy-Schwarz inequality, the second term satisfies
\beas
\frac{1}{N^3} \sum_{q,k} \left|
\sum_{l} \omega^{-kl}
\sum_{m\neq l} \omega^{km} U^l_{m,q}  \right|^2 
&\leq& \frac{1}{N^2} \sum_{q,k} 
\sum_{l} \left|\sum_{m\neq l} \omega^{km} U^l_{m,q}  \right|^2 \\
&=& \frac{1}{N^2} \sum_{q,k} 
\sum_{l}  \sum_{m_1,m_2\neq l} \omega^{k(m_1-m_2)} U^l_{m_1,q} \overline{U}^l_{m_2,q} \\
&=& \frac{1}{N} \sum_{q} 
\sum_{l}  \sum_{m\neq l}  |U^l_{m,q}|^2 \\
&=& \frac{1}{N} 
\sum_{l} \sum_{q}  \left(\sum_{m}  |U^l_{m,q}|^2 - |U^l_{l,q}|^2 \right).
\eeas
Since $U$ is a unitary, $\sum_{q}\sum_{m}  |U^l_{m,q}|^2=1$, so the above equals
\[
\frac{1}{N} 
\sum_{l}   \left(1 - \sum_{q} |U^l_{l,q}|^2 \right) 
= 1 -  \frac{1}{N} \sum_{l} \sum_{q} |U^l_{l,q}|^2 .
\]
By the first inequality of (\ref{ineqs}), the above is upper bounded by $\eta_1$.
From the above calculation, we finally obtain
\[
\E_{k,l} \bra{k} \, C(\ket{\hat{k}}\bra{\hat{l}}) \, \ket{l} = 
\frac{1}{N^2} \sum_{q} \left|\sum_{l} U^l_{l,q}\right|^2 \geq 1-\eta_1-\eta_2.
\]
This means $C\in S_3(\eta_1+\eta_2)$.

Next we prove that if $C\in S_3(\eta)$, then $C \in S_1(\eta) \cap S_2(\eta)$.
Similar to the computation of (\ref{0204eq1}), we have that $C\in S_1(\eta)$ if and only if 
\bes
1-\eta \leq \sum_{i=1}^r \E_{k} \bra{k} A_i \ket{\hat{k}} \bra{\hat{k}} A_i^\dag \ket{k}
= \sum_{i=1}^r  \E_k[|\bra{k} A_i \ket{\hat{k}}|^2].
\ees
As a result, we have $S_3(\eta) \subseteq S_1(\eta)$. 
About $S_2(\eta)$ we have $C\in S_2(\eta)$ if and only if 
\bes
1-\eta \leq \sum_{i=1}^r \E_{k} \bra{\widehat{-k}} A_i \ket{k} \bra{k} A_i^\dag \ket{\widehat{-k}}
= \sum_{i=1}^r  \E_k[|\bra{\widehat{-k}} A_i \ket{k}|^2].
\ees
Together with Proposition \ref{prop:fact2 of S3}, we have $S_3(\eta) \subseteq S_2(\eta)$.
\end{proof}

\section{Similar results for verifying correctness of $F_N$ on average}
\label{section:Similar results for QFT}

In the HHL algorithm, we also need to use $F_N$. With the above preliminaries, we can propose similar protocols for verifying if a quantum channel $P$ is close to $F_N$ on average. To this end, we similarly define
\beas
T_1(\eta) &:=& \Big\{P \text{ is a quantum channel}: 
\E_k\Big[ \F\Big(
P(\ket{\hat{k}}\bra{\hat{k}}), \ket{-k} \bra{-k}
\Big) \Big] \geq 1-\eta \Big\}, \\
T_2(\eta) &:=& \Big\{P \text{ is a quantum channel}: \E_k\Big[\F\Big(P (\ket{k} \bra{k}), \ket{\hat{k}} \bra{\hat{k}} \Big)\Big] \geq 1-\eta \Big\}, \\
T_3(\eta) &:=& \Big\{P \text{ is a quantum channel}: \E_{k,l} \bra{\hat{k}} \, P(\ket{k} \bra{l}) \, \ket{\hat{l}}  \geq 1- \eta \Big\}.
\eeas

\begin{prop}
\label{prop:0226}
If $P \in T_1(\eta_1) \cap T_2(\eta_2)$, then $P\in T_3(\eta_1+\eta_2)$. Conversely, if $P\in T_3(\eta)$, then $P\in T_1(\eta) \cap T_2(\eta)$.
\end{prop}

\begin{proof}
This follows from an argument similar to the proof of Theorem \ref{thm:testing S3}.
\end{proof}

The next result states that if we only have a quantum channel $C$ that is close to $F_N^{-1}$ on average, then we can use it to find a quantum channel $P$ close to $F_N$ on average. 

\begin{prop}
Assume that $C$ is a quantum channel.
\begin{enumerate}
\item Let $R$ be the unitary channel such that $R (\ket{k} \bra{l}) = \ket{-k} \bra{-l}$. For $i\in \{1,2\}$, we have that
if $C \in S_i(\eta)$, then $C \circ R, R \circ C\in T_i(\eta)$.
%
%
\item If $C \in S_1(\eta_1) \cap S_2(\eta_2)$, then $C^3 \in T_1(\sqrt{\eta_1} + \sqrt{\sqrt{\eta_1} + \sqrt{\eta_2}}) \cap T_2(\sqrt{\eta_2} + \sqrt{\sqrt{\eta_1} + \sqrt{\eta_2}})$.
Consequently, if $C\in S_3(\eta)$, then $C^3 \in T_3(2\eta^{1/2}+2\sqrt{2}\eta^{1/4})$.
\end{enumerate}
\end{prop}

\begin{proof}
(1). Suppose $C\in S_1(\eta)$. It is easy to check that $R(\ket{\hat{k}}\bra{\hat{k}}) = \ket{\widehat{-k}}\bra{\widehat{-k}}$,
so
\beas
\E_k[\F(C\circ R (\ket{\hat{k}}\bra{\hat{k}}) , \, \ket{-k}\bra{-k})] 
&=&
\E_k[\F(C (\ket{\widehat{-k}}\bra{\widehat{-k}}) , \, \ket{-k}\bra{-k} )] \\
&=&
\E_k[\F(C (\ket{\hat{k}}\bra{\hat{k}}) , \, \ket{k}\bra{k})] \\
&\geq& 
1-\eta.
\eeas
Similarly, we have
\beas
\E_k[\F(R\circ C (\ket{\hat{k}}\bra{\hat{k}}) , \, \ket{-k}\bra{-k})] 
&=&
\E_k[\F(R\circ C (\ket{\hat{k}}\bra{\hat{k}}) , \, R (\ket{k}\bra{k}) )] \\
&\geq&
\E_k[\F(C (\ket{\hat{k}}\bra{\hat{k}}) , \, \ket{k}\bra{k})] \\
&\geq& 
1-\eta.
\eeas
So $C \circ R, R \circ C\in T_1(\eta)$. The argument for $T_2(\eta)$ is similar.



(2). First, we show that $C^3 \in T_1(\sqrt{\eta_1} + \sqrt{\sqrt{\eta_1} + \sqrt{\eta_2}})$. 
By Proposition \ref{prop:inequality for fidelity}, we have
\beas
&& \F\Big(
C^3(\ket{\hat{k}}\bra{\hat{k}}), \ket{-k} \bra{-k}
\Big)\\
&\geq&
1-
\sqrt{1-\F\Big(
C^3(\ket{\hat{k}}\bra{\hat{k}}), C^2(\ket{k} \bra{k})
\Big)  }
-
\sqrt{1-\F\Big(
C^2(\ket{k} \bra{k}), \ket{-k} \bra{-k}
\Big) } \\
&\geq&
1-
\sqrt{1-\F\Big(
C(\ket{\hat{k}}\bra{\hat{k}}), \ket{k} \bra{k}
\Big)  }
-
\sqrt{1-\F\Big(
C^2(\ket{k} \bra{k}), \ket{-k} \bra{-k}
\Big) }.
\eeas
By Proposition \ref{prop:inequality for fidelity} again, we have
\beas
&& \F\Big(
C^2(\ket{k} \bra{k}), \ket{-k} \bra{-k}
\Big) \nonumber \\
&\geq& 
1 - 
\sqrt{1 - \F\Big(
C^2(\ket{k} \bra{k}), C(\ket{\widehat{-k}} \bra{\widehat{-k}})
\Big) }
-
\sqrt{1 - \F\Big(
C(\ket{\widehat{-k}} \bra{\widehat{-k}}), \ket{-k} \bra{-k}
\Big)  } \nonumber \\
&\geq& 
1 - 
\sqrt{1 - \F\Big(
C(\ket{k} \bra{k}), \ket{\widehat{-k}} \bra{\widehat{-k}}
\Big) }
-
\sqrt{1 - \F\Big(
C(\ket{\widehat{-k}} \bra{\widehat{-k}}), \ket{-k} \bra{-k}
\Big)  } .
\eeas
Since $C \in S_1(\eta_1) \cap S_2(\eta_2)$, we have
\be
\F\Big(
C^2(\ket{k} \bra{k}), \ket{-k} \bra{-k}
\Big) \geq 1 - \sqrt{\eta_1} - \sqrt{\eta_2}.
\label{0225}
\ee
Thus
\[
\F\Big(
C^3(\ket{\hat{k}}\bra{\hat{k}}), \ket{-k} \bra{-k}
\Big) 
\geq 1 - \sqrt{\eta_1} - \sqrt{\sqrt{\eta_1} + \sqrt{\eta_2}}.
\]

Next, we show that $C^3 \in T_2(\sqrt{\eta_2} + \sqrt{\sqrt{\eta_1} + \sqrt{\eta_2}})$. The proof is similar. By Proposition \ref{prop:inequality for fidelity}, we have
\beas
&& \F\Big(
C^3(\ket{k}\bra{k}), \ket{\hat{k}} \bra{\hat{k}}
\Big)\\
&\geq&
1-
\sqrt{1-\F\Big(
C^3(\ket{k}\bra{k}), C(\ket{-k} \bra{-k})
\Big) }
-
\sqrt{1-\F\Big(
C(\ket{-k} \bra{-k}), \ket{\hat{k}} \bra{\hat{k}}
\Big) } \\
&\geq&
1-
\sqrt{1-\F\Big(
C^2(\ket{k}\bra{k}), \ket{-k} \bra{-k}
\Big) }
-
\sqrt{1-\F\Big(
C(\ket{k} \bra{k}), \ket{\widehat{-k}} \bra{\widehat{-k}}
\Big) }.
\eeas
By (\ref{0225}) and the assumption that $C\in S_2(\eta_2)$, we have
\[
\F\Big(
C^3(\ket{k}\bra{k}), \ket{\hat{k}} \bra{\hat{k}}
\Big)
\geq 1 - \sqrt{\eta_2} - \sqrt{\sqrt{\eta_1} + \sqrt{\eta_2}}.
\]

The last claim now follows from Theorem \ref{thm:testing S3} and Proposition \ref{prop:0226}.
\end{proof}

\section{Worst-case HHL with average-case-correct QFT}
\label{sectionn:Worst-case HHL with average-case-correct QFT}

In this section, we analyze the performance of running the HHL algorithm using average-case-correct QFT. A key step of the HHL algorithm is quantum phase estimation. So, similar to \cite{linden2021average}, we first focus on the perfect case where all the eigenvalues of the input matrix are represented by $n$ bits precisely. We then extend the proofs and results to the general case.

\subsection{Outline of the HHL algorithm}
\label{subsection:Outline of HHL}

Let $A\in \mathbb{C}^{d\times d}$ be a sparse Hermitian matrix, $\b\in \mathbb{C}^{d}$ be a vector. Suppose we have the ability to prepare the quantum state $\ket{\b}$ efficiently. On a quantum computer, we can use the HHL algorithm to solve the linear system $A\x=\b$. The output is a quantum state close to $\ket{A^+\b}$, where $A^+$ is the pseudoinverse. Moreover, using a similar idea to HHL, we can prepare the state $\ket{f(A)\b}$ for any efficiently computable function $f(x)$, and the HHL algorithm corresponds to $f(x)=1/x$. For completeness, below we outline the procedure of preparing $\ket{f(A)\b}$ for a general $f(x)$.

Let the eigenvalue decomposition of $A$ be $A = \sum_k \sigma_k \ket{\u_k} \bra{\u_k}$. Then there exist $\beta_k\in \mathbb{C}$ such that $\ket{\b} = \sum_k \beta_k \ket{\u_k}$. To better understand the analysis, below we only state the algorithm in the perfect case by assuming that $\sigma_k = p_k/N$ for some integer $p_k$, where $N=2^n$ is determined by the accuracy. Namely, $\{\sigma_k:k\in [d]\}$ are $n$-bit numbers exactly. Without loss of generality, we assume that $f(x)$ is bounded by 1 in the interval $[-1,1]$. The following is a basic description of the HHL algorithm:

\begin{enumerate}
\item Prepare 
\[
\ket{\phi_1} = \frac{1}{\sqrt{N}} \sum_{j\in[N]} \ket{j} \ket{\b}
= \frac{1}{\sqrt{N}} \sum_{j\in[N]} \sum_{k\in [d]} \beta_k \ket{j} \ket{\u_k}.
\]
\item Apply $\sum_j \ket{j} \bra{j} \otimes e^{2\pi i A j}$
\[
\ket{\phi_2} = \frac{1}{\sqrt{N}} \sum_{j\in[N]} \sum_{k\in [d]} \beta_k e^{2\pi i j \sigma_k} \ket{j} \ket{\u_k}
= \sum_{k\in [d]} \beta_k \ket{\widehat{p_k}} \ket{\u_k}.
\]
\item Apply $F_N^{-1}\otimes I$
\[
\ket{\phi_3} = \sum_{k\in [d]} \beta_k \ket{p_k} \ket{\u_k}.
\]
\item Apply control rotations
\[
\ket{\phi_4} = \sum_{k\in [d]} \beta_k \ket{p_k} \ket{\u_k}\otimes \left(f(\sigma_k) \ket{0} + \sqrt{1-f(\sigma_k)^2} \ket{1}\right).
\]
\item Apply $F_N  \otimes I$
\beas
\ket{\phi_5} & = &\sum_{k\in [d]} \beta_k \ket{\widehat{p_k}} \ket{\u_k}\otimes \left(f(\sigma_k) \ket{0} + \sqrt{1-f(\sigma_k)^2} \ket{1}\right) \\
&= & \frac{1}{\sqrt{N}} \sum_{j\in[N]} \sum_{k\in [d]}  \beta_k e^{2\pi i j \sigma_k} \ket{j} \ket{\u_k}\otimes \left(f(\sigma_k) \ket{0} + \sqrt{1-f(\sigma_k)^2} \ket{1}\right).
\eeas
\item Apply $\sum_j \ket{j} \bra{j} \otimes e^{-2\pi i A j}$
\beas
\ket{\phi_6} 
= \frac{1}{\sqrt{N}} \sum_{j\in[N]}\ket{j}  \otimes \sum_{k\in [d]}  \beta_k  \ket{\u_k}\otimes \left(f(\sigma_k) \ket{0} + \sqrt{1-f(\sigma_k)^2} \ket{1}\right) .
\eeas
\item Apply Hadamard gates to $\ket{j}$
\bea
\ket{\phi_{7}} 
&=& \ket{0}^{\otimes n}  \otimes \sum_{k\in [d]}  \beta_k  \ket{\u_k}\otimes \left(f(\sigma_k) \ket{0} + \sqrt{1-f(\sigma_k)^2} \ket{1}\right) \nonumber \\
&=& \ket{0}^{\otimes n}  \otimes f(A) \ket{\b}\otimes \ket{0} + \ket{0}^{\otimes n}  \otimes \sqrt{1-f^2(A)} \ket{\b} \otimes \ket{1} .
\label{resulting state}
\eea
\end{enumerate}

The information of $f(A) \b$ is now included in the first term of $\ket{\phi_{7}}$. To obtain its quantum state, it suffices to measure the last qubit. We will obtain $\ket{f(A) \b}$ if the output of the measurement is 0.
We can also use amplitude amplification to increase the success probability, but that is a much less lightweight process that involves many applications of $\ket{\phi_{7}}$. Note that one application of the HHL algorithm is to estimate the expectation value $\langle \x|M|\x\rangle$ for some operator $M$, where $\ket{\x}= \ket{f(A) \b}$. We can also achieve this goal using $\ket{\phi_{7}}$. Namely, we compute $\bra{\phi_7} (I^n \otimes M \otimes \ket{0} \bra{0}) \ket{\phi_7}$. Because of the above reasons, in the following, we only focus on the analysis for $\ket{\phi_{7}}$. 

\subsection{Main result in the case when eigenvalues are represented exactly in $ n$ bits}

To introduce randomness, we shall use the same idea given in \cite{linden2021average} by considering $A_l := A + l I_d/N$ for a random $l\in[N]$, which is chosen uniformly at random. The eigenvalues of $A_l$ are $(p_k+l)/N$, where $k\in[d]$. The eigenvectors do not change. From $A_l$, we can also construct the state $\ket{f(A) \b}$. It suffices to change $f(x)$ to $f_l(x) := f(x-l/N)$. We denote the final state as $\ket{\phi_{7,l}}$. It is $\ket{\phi_{7}}$ in the case where all eigenvalues are represented by $n$-bit. But generally, it may depend on $l$.

Given $C\in S_3(\eta_1), P\in T_3(\eta_2)$, we can also run the above procedure (i.e., step 1 to step 7) by replacing $F_N^{-1}$ with $C$ and replacing $F_N$ with $P$. 
Since $C, P$ are quantum channels, the algorithm should be implemented in the form of density operators.
Thus, we denote the final state as $\rho_{7,l}$. 
When running HHL, we assume that all unitaries except $F_N, F_{N}^{-1}$ are perfect.
Our main result in the case when all eigenvalues of $A$ are represented exactly by $n$ bits is as follows.

\begin{thm}
\label{main theorem}
Let $A\in \mathbb{C}^{d\times d}$ be a sparse Hermitian matrix with $\|A\|\leq 1$, $\ket{\b}\in \mathbb{C}^{d}$ be a given state and $f(x)$ be a function bounded by 1 when $x\in[-1,1]$. Let $\ket{\phi_{7,l}} = \ket{0}\otimes f(A)\ket{\b} + \ket{0^\bot}$ be the resulting state (\ref{resulting state}) obtained by the HHL algorithm.\footnote{Here we simplified the notation of $\ket{\phi_{7,l}}$. In the setting here, it is indeed independent of $l$ and so equals $\ket{\phi_7}$.} Assume that $C\in S_3(\eta_1), P\in T_3(\eta_2)$, and the eigenvalues of $A$ are represented by $n$ bits precisely. For any $l\in [N]$, let $\rho_{7,l}$ be the resulting state by replacing $F_N^{-1}$ with $C$ and replacing $F_N$ with $P$ in the HHL algorithm with inputs $A + l I_d/N, \ket{\b}$ and $f(x-l/N)$.
Then
\be
\E_l\Big[ \F\Big( \rho_{7,l} \, , \, \ket{\phi_{7,l}} \bra{\phi_{7,l}} \Big) \Big]
\geq 
1- \sqrt{\eta_1} - \sqrt{\eta_2}.
\ee
\end{thm}

\begin{proof}
Because of the introduction of $l$, we denote the states appeared in the HHL algorithm as $\ket{\phi_{1,l}},\ldots,\ket{\phi_{7,l}}$. 
Since all unitaries except $F_N, F_{N}^{-1}$to are perfect, it suffices to focus on the changes on the states $\ket{\phi_{2,l}},\ldots,\ket{\phi_{5,l}}$.

By assumption, 
\[
\ket{\phi_{2,l}} = \sum_{k\in [d]} \beta_k \ket{\widehat{p_k+l}} \ket{\u_k},
\quad 
\ket{\phi_{3,l}} = \sum_{k\in[d]} \beta_k \ket{p_k+l} \ket{\u_k}.
\]
We can check that
\[
C\otimes I (\ket{\phi_{2,l}}\bra{\phi_{2,l}})
= \sum_{k_1,k_2\in [d]} \beta_{k_1} \bar{\beta}_{k_2} C\Big( \ket{\widehat{p_{k_1}+l}} \bra{\widehat{p_{k_2}+l}} \Big)  \otimes \ket{\u_{k_1}} \bra{\u_{k_2}}.
\]
Since $C$ is a quantum channel, there exist Kraus operators $A_1,\ldots,A_r$ satisfying $\sum_i A_i^\dag A_i = I$ and $C(\rho) = \sum_i A_i \rho A_i^\dag$ for any density operator $\rho$.
So we have
\be
1-\eta_1 \leq \E_{k,l} \bra{k} C(\ket{\hat{k}} \bra{\hat{l}}) \ket{l}
=\sum_{i=1}^r \E_{k,l} \bra{k} A_i \ket{\hat{k}} \bra{\hat{l}} A_i^\dag \ket{l}
=\sum_{i=1}^r |\E_k \bra{k} A_i \ket{\hat{k}} |^2.
\label{0226eq1}
\ee
Thus
\beas
&& \E_l\left[\F\Big(C\otimes I (\ket{\phi_{2,l}}\bra{\phi_{2,l}}), \ket{\phi_{3,l}}\bra{\phi_{3,l}} \Big)\right] \\
&=& \E_l \sum_{k_1,k_2\in [d]} |\beta_{k_1}|^2 |\beta_{k_2}|^2 \langle p_{k_1}+l| C\Big( \ket{\widehat{p_{k_1}+l}} \bra{\widehat{p_{k_2}+l}} \Big)  |p_{k_2}+l \rangle \\
&=& \sum_{i=1}^r \E_l \sum_{k_1,k_2\in [d]} |\beta_{k_1}|^2 |\beta_{k_2}|^2 \langle p_{k_1}+l| A_i \ket{\widehat{p_{k_1}+l}} \bra{\widehat{p_{k_2}+l}} A_i^\dag |p_{k_2}+l \rangle \\
&=& \sum_{i=1}^r \E_l \left[\left|\sum_{k\in [d]} |\beta_{k}|^2 \langle p_{k}+l| A_i \ket{\widehat{p_{k}+l}}\right|^2\right].
\eeas
By Jensen's inequality (see Fact \ref{fact:Jensen's inequality}), it is bounded from below by
\beas
&\geq& \sum_{i=1}^r \left|\E_l \left[\sum_{k\in [d]} |\beta_{k}|^2 \langle p_{k}+l| A_i \ket{\widehat{p_{k}+l}}\right] \right|^2\\
&=& \sum_{i=1}^r \left| \sum_{k\in [d]} |\beta_{k}|^2 \E_l \langle l| A_i \ket{\hat{l}} \right|^2 \\
&=& \sum_{i=1}^r \left| \E_l \langle l| A_i \ket{\hat{l}} \right|^2.
\eeas
By (\ref{0226eq1}), it is greater than $1-\eta_1$.

Denote the control rotation that maps $\ket{\phi_{3,l}}$ to $\ket{\phi_{4,l}}$ as $R_l$, which contains no error by assumption.
Note that
\[
\ket{\phi_{4,l}}
= \sum_{k\in [d]} \beta_k \ket{p_k+l} \ket{\u_k}\otimes \left(f(\sigma_k) \ket{0} + \sqrt{1-f(\sigma_k)^2} \ket{1}\right)
\]
and applying the quantum channel only increases fidelity,
it follows that
\[
\E_l\left[\F\Big(R_l(C\otimes I) (\ket{\phi_{2,l}}\bra{\phi_{2,l}}), \ket{\phi_{4,l}}\bra{\phi_{4,l}} \Big)\right]
\geq 
\E_l\left[\F\Big(C\otimes I(\ket{\phi_{2,l}}\bra{\phi_{2,l}}), \ket{\phi_{3,l}}\bra{\phi_{3,l}} \Big)\right]
\geq 1-\eta_1.
\]

Since
\[
\ket{\phi_{5,l}} = \sum_{k\in [d]} \beta_k \ket{\widehat{p_k+l}} \ket{\u_k}\otimes \left(f(\sigma_k) \ket{0} + \sqrt{1-f(\sigma_k)^2} \ket{1}\right),
\]
similar to the above estimation, we have
\[
\E_l\left[\F\Big(P\otimes I(\ket{\phi_{4,l}}\bra{\phi_{4,l}}), \, \ket{\phi_{5,l}}\bra{\phi_{5,l}} \Big)\right]
\geq 1-\eta_2.
\]
Therefore, by Proposition \ref{prop:inequality for fidelity} and Jensen's inequality, we have
\beas
&& \E_l\left[\F\Big((P\otimes I)R_l(C\otimes I) (\ket{\phi_{2,l}}\bra{\phi_{2,l}}), \ket{\phi_{5,l}}\bra{\phi_{5,l}} \Big)\right] \\
&\geq& 1- \sqrt{1-\E_l\F\Big(P\otimes I(\ket{\phi_{4,l}}\bra{\phi_{4,l}}), \ket{\phi_{5,l}}\bra{\phi_{5,l}} \Big)} \\
&& -\, \sqrt{1-\E_l\F\Big(P\otimes I(\ket{\phi_{4,l}}\bra{\phi_{4,l}}), (P\otimes I)R_l(C\otimes I) (\ket{\phi_{2,l}}\bra{\phi_{2,l}}) \Big)} \\
&\geq& 1- \sqrt{1-\E_l\F\Big(P\otimes I(\ket{\phi_{4,l}}\bra{\phi_{4,l}}), \ket{\phi_{5,l}}\bra{\phi_{5,l}} \Big)} \\
&& -\, \sqrt{1-\E_l\F\Big((\ket{\phi_{4,l}}\bra{\phi_{4,l}}), R_l(C\otimes I) (\ket{\phi_{2,l}}\bra{\phi_{2,l}}) \Big)} \\
&\geq& 1- \sqrt{\eta_1} - \sqrt{\eta_2}.
\eeas
The claimed result follows from the above estimate since the unitaries in steps 6 and 7 are also perfect with no error by assumption and can only increase the fidelity.
\end{proof}



One application of the HHL algorithm is to compute the expectation value
$\bra{\b f(A)^T} M \ket{f(A)\b}$ for some operator $M$, or more generally $\langle \phi_{7}|M|\phi_{7}\rangle$ in our notation. The above result suggests that on average, we can use ${\rm Tr}(M \rho_{7,l})$ to approximate the desired value. We state this result as follows.

\begin{thm}
\label{thm2 for perfect case}
Making the same assumption and notation as Theorem \ref{main theorem}, let $M$ be any operator with $\|M\|_2\leq 1$, denote 
$X_l = {\rm Tr}(M \ket{\phi_{7,l}} \bra{\phi_{7,l}}) , 
\widetilde{X}_l = {\rm Tr}(M \rho_{7,l}) $, then
\[
\E_l[|\widetilde{X}_l - X_l|] \leq 2 (\eta_1^{1/4}+\eta_2^{1/4}).
\]
\end{thm}
\begin{proof}
By von Neumann's trace inequality (see Fact \ref{fact:von Neumann's trace inequality}), for any two operators $A,B$ we have ${\rm Tr}(AB) \leq \|A\|_2 \|B\|_{1}$, where $\|\cdot\|_2$ is the operator norm and $\|\cdot\|_{1}$ is the trace norm. So
by (\ref{ineq1}) and Theorem \ref{main theorem}
we have
\beas
\E_l[|\widetilde{X}_l - X_l|]
&=& \E_l[|{\rm Tr}(M (\rho_{7,l}-\ket{\phi_{7,l}} \bra{\phi_{7,l}}))|] \\
&\leq& \|M\|_2 \E_l[\|\rho_{7,l}-\ket{\phi_{7,l}} \bra{\phi_{7,l}}\|_{1}] \\
&=& 2 \|M\|_2 \E_l[\T(\rho_{7,l},\ket{\phi_{7,l}} \bra{\phi_{7,l}})] \\
&\leq& 2 \sqrt{1-\E_l[\F(\rho_{7,l},\ket{\phi_{7,l}} \bra{\phi_{7,l}})]} \\
&\leq& 2 \sqrt{\sqrt{\eta_1}+\sqrt{\eta_2}} \\
&\leq& 2 (\eta_1^{1/4}+\eta_2^{1/4}).
\eeas
In the second equality, we used the definition of trace distance. 
\end{proof}


\subsection{Main result in the general case}

In the general case, the eigenvalues of $A$ may be specified by more than $n$ bits. The analysis here will be a little complicated. Recall that $N=2^n$. For each $\sigma_k$, denote $p_k = \lfloor \sigma_k N \rfloor$, i.e., $\sigma_k = p_k/N + \varepsilon$ for some error $\varepsilon$. Denote
\be
\label{good sets}
G_k = \{p_k-K+1,\ldots,p_k+K\},
\ee
which stands for the set of good candidates that can be used to approximate $\sigma_k$. Then $|G_k| = 2K$.
Let
\be
\label{eq for general case}
\ket{\psi_k}  = \frac{1}{\sqrt{N}} \sum_{j\in[N]} e^{2\pi i j \sigma_k} \ket{j} 
= \sum_{g \in G_k} \alpha_{kg} \ket{\hat{g}} + \sum_{g \notin G_k} \alpha_{kg} \ket{\hat{g}}
\ee
for some $\alpha_{kg}\in\mathbb{C}$.
As shown in \cite[Proposition 2]{linden2021average}, for any constant $K>1$, $\sum_{g \notin G_k} |\alpha_{kg}|^2 \leq 2/(K-1)$. This is negligible when $K$ is large enough.

\begin{lem}
\label{lem1:general case}
Let $C\in S_3(\eta)$ be a quantum channel with Kraus representation $C(\rho) = \sum_{i=1}^r A_i \rho A_i^\dag$ for any density operator $\rho$. Let $G \subseteq [N]$ be a subset and
\beas
\ket{\psi_l} = \sum_{g \in G} \alpha_{g} \ket{\widehat{g+l}} + \sum_{g \notin G} \alpha_{g} \ket{\widehat{g+l}}, \quad
\ket{\varphi_l} = \sum_{g \in G} \alpha_{g} \ket{{g+l}} + \sum_{g \notin G} \alpha_{g} \ket{{g+l}}.
\eeas
Assume that $\sum_{g\in G}|\alpha_g|^2 = 1-\delta$, then
\[
\E_l\Big[\bra{\varphi_l} A_i \ket{\psi_l} \Big]
= (1-\delta) \E_l\Big[\bra{l} A_i \ket{\hat{l}\,} \Big]
+ Err_i,
\]
for some error term $Err_i$. Moreover, $\sum_{i=1}^r |Err_i|^2 \leq 2 \eta |G|^2 + 18 \delta.$
\end{lem}

\begin{proof}
For convenience, we denote $\ket{\psi_l} = \sqrt{1-\delta} \, \ket{\psi_{l1}} + \sqrt{\delta} \, \ket{\psi_{l2}},
\ket{\varphi_l} = \sqrt{1-\delta} \, \ket{\varphi_{l1}} + \sqrt{\delta} \, \ket{\varphi_{l2}}$ according to their decompositions. Then
\beas
\E_l\Big[\bra{\varphi_l} A_i \ket{\psi_l} \Big]
&=& (1-\delta)\E_l\Big[\bra{\varphi_{l1}} A_i \ket{\psi_{l1}} \Big] + Err_{i1},
\eeas
where $Err_{i1} = \sqrt{\delta(1-\delta)} 
\E_l\bra{\psi_{l1}} A_i \ket{\varphi_{l2}} + 
\sqrt{\delta(1-\delta)} 
\E_l\bra{\psi_{l2}} A_i \ket{\varphi_{l1}} +
\delta \E_l\bra{\psi_{l2}} A_i \ket{\varphi_{l2}}$.

We can compute that
\beas
(1-\delta)\E_l\Big[\bra{\varphi_{l1}} A_i \ket{\psi_{l1}} \Big] 
&=& \sum_{g_1,g_2\in G} \bar{\alpha}_{g_1} \alpha_{g_2} \E_l \bra{g_1+l} A_i \ket{\widehat{g_2+l}} \\
&=& \sum_{g\in G} |\alpha_{g}|^2 \E_l \bra{l} A_i \ket{\hat{l}} 
+ Err_{i2} \\
&=& (1-\delta) \E_l \bra{l} A_i \ket{\hat{l}} 
+ Err_{i2},
\eeas
where
\[
Err_{i2} = \sum_{g_1\neq g_2\in G} \bar{\alpha}_{g_1} \alpha_{g_2} \E_l \bra{g_1+l} A_i \ket{\widehat{g_2+l}}.
\]
The overall error is $Err_i := Err_{i1} +Err_{i2}$.

From the Kraus representation of $C$, it is easy to check that for any states $\ket{a_l}, \ket{b_l}$, we have 
\be
\E_l\left[\sum_{i=1}^r |\bra{a_l} A_i \ket{b_l}|^2 \right]
=\E_l \bra{a_l}\, C(\ket{b_l} \bra{b_l}) \, \ket{a_l}.
\label{0225eq1}
\ee
So we have
\beas
\sum_{i=1}^r |Err_{i1}|^2 &\leq& 3\E_l\sum_{i=1}^r 
\Big( \delta(1-\delta) |\bra{\psi_{l1}} A_i \ket{\varphi_{l2}}|^2 + \delta(1-\delta) |\bra{\psi_{l2}} A_i \ket{\varphi_{l1}}|^2 +\delta^2 |\bra{\psi_{l2}} A_i \ket{\varphi_{l2}}|^2 \Big) \\
&=& 3 \delta(1-\delta) \E_{l} \bra{\psi_{l1}} C (\ket{\varphi_{l2}} \bra{\varphi_{l2}} ) \ket{\psi_{l1}} 
+
3 \delta(1-\delta) \E_{l} \bra{\psi_{l2}} C (\ket{\varphi_{l1}} \bra{\varphi_{l1}} ) \ket{\psi_{l2}} \\
&& +\, 3 \delta^2 
\E_{l} \bra{\psi_{l2}} C (\ket{\varphi_{l2}} \bra{\varphi_{l2}} ) \ket{\psi_{l2}} \\
&\leq& 6 \delta(1-\delta) + 3\delta^2 \leq 9\delta.
\eeas
And
\beas
\sum_{i=1}^r |Err_{i2}|^2
&=& \sum_{i=1}^r 
\left|\sum_{g_1\neq g_2\in G} \bar{\alpha}_{g_1} \alpha_{g_2} \E_l \bra{g_1+l} A_i \ket{\widehat{g_2+l}} \right|^2 \\
&\leq& \E_l \left[\sum_{i=1}^r 
\left|\sum_{g_1\neq g_2  \in G} \bar{\alpha}_{g_1} \alpha_{g_2} 
\bra{g_1+l} \, A_i \ket{\widehat{g_2+l}} \right|^2 \right]\\
&=& 
 \sum_{g_1\neq g_2, g_3\neq g_4  \in G}
 \bar{\alpha}_{g_1} \alpha_{g_2} 
 \alpha_{g_3} \bar{\alpha}_{g_4} 
\E_l \left[\sum_{i=1}^r \bra{g_1+l} \, A_i \ket{\widehat{g_2+l}}  
\bra{\widehat{g_3+l}} \, A_i^\dag \ket{g_4+l} \right].
\eeas
By the Cauchy-Schwarz inequality, the above is smaller than
\[
\leq \sum_{g_1\neq g_2, g_3\neq g_4  \in G}
 |\bar{\alpha}_{g_1} \alpha_{g_2} 
 \alpha_{g_3} \bar{\alpha}_{g_4}| 
 \sqrt{\E_l\left[\sum_{i=1}^r |\bra{g_1+l} \, A_i \ket{\widehat{g_2+l}}|^2 \right] }
 \sqrt{\E_l\left[\sum_{i=1}^r |\bra{\widehat{g_3+l}} \, A_i^\dag \ket{g_4+l}|^2\right]}.
\]
By (\ref{0225eq1}), it is
\beas
&=& \sum_{g_1\neq g_2, g_3\neq g_4  \in G}
 |\bar{\alpha}_{g_1} \alpha_{g_2} 
 \alpha_{g_3} \bar{\alpha}_{g_4}| 
 \sqrt{\E_l \bra{g_1+l} \, C( \ket{\widehat{g_2+l}} \bra{\widehat{g_2+l}}) \, \ket{g_1+l}} \\
 && \hspace{2cm} \times \,
 \sqrt{\E_l \bra{g_4+l} \, C( \ket{\widehat{g_3+l}} \bra{\widehat{g_3+l}}) \, \ket{g_4+l}}.
\eeas
By Proposition \ref{prop:nondiagonal entries of S3}, it is smaller than
\[
\leq 
\eta   \sum_{g_1\neq g_2, g_3\neq g_4  \in G}
 |\bar{\alpha}_{g_1} \alpha_{g_2} 
 \alpha_{g_3} \bar{\alpha}_{g_4}| 
\leq \eta |G|^2 (1-\delta)^2   \leq  \eta |G|^2.
\]

Finally, we have
\[
\sum_{i=1}^r |Err_{i}|^2\leq 2 \sum_{i=1}^r |Err_{i1}|^2 + 2 \sum_{i=1}^r |Err_{i2}|^2
\leq 2\eta |G|^2 + 18 \delta.
\]
This completes the proof.
\end{proof}

\begin{thm}
\label{main theorem:general case}
Using the same notation as Theorem \ref{main theorem}. Assume that $C\in S_3(\eta_1), P\in T_3(\eta_2)$. 
Then for any integer $K\geq 1$,
\be
\E_l\Big[\F\Big( \rho_{7,l}, \, \ket{\phi_{7,l}} \bra{\phi_{7,l}} \Big)\Big]
\geq 
1 - (\eta_1^{1/2}+\eta_2^{1/2})
- 2 \sqrt{K} (\eta_1^{1/4}+\eta_2^{1/4})
- \frac{4\sqrt{5}}{(K-1)^{1/4}}.
\ee
\end{thm}

\begin{proof}
As before, we introduce a random variable $l\in[N]$ and denote
\[
\ket{\psi_{kl}} = \frac{1}{\sqrt{N}} \sum_{j\in[N]} e^{2\pi i j (\sigma_k+l/N)} \ket{j} 
= \sum_{g \in G_k} \alpha_{kg} \ket{\widehat{g+l}} + \sum_{g \notin G_k} \alpha_{kg} \ket{\widehat{g+l}},
\]
where $G_k$ is defined in (\ref{good sets}).
In the general case, the state in step 2 of the HHL algorithm (see Subsection \ref{subsection:Outline of HHL}) becomes
$
\ket{\phi_2} = 
\sum_{k\in [d]} \beta_k
\ket{\psi_k}
\ket{\u_k},
$
where $\ket{\psi_k}$ is of the form (\ref{eq for general case}).
When introducing a random variable $l$, we obtain
$
\ket{\phi_{2,l}} = 
\sum_{k\in [d]} \beta_k
\ket{\psi_{kl}}
\ket{\u_k}.
$
By applying $C\in S_3(\eta_1)$ to it, we obtain a state that is close to
$
\ket{\phi_{3,l}} = 
\sum_{k\in [d]} \beta_k \ket{\varphi_{kl}}
\ket{\u_k},
$
where
\[
\ket{\varphi_{kl}} = \sum_{g \in G_k} \alpha_{kg} \ket{g+l} + \sum_{g \notin G_k} \alpha_{kg} \ket{g+l}.
\]
We now aim to estimate
$\E_l[\F(C\otimes I (\ket{\phi_{2,l}} \bra{\phi_{2,l}}), \ket{\phi_{3,l}} \bra{\phi_{3,l}}) ] $. Similar to the proof of Theorem \ref{main theorem}, we let $C(\rho) = \sum_{i=1}^r A_i \rho A_i^\dag$ be the Kraus representation of $C$, where $\sum_{i=1}^r A_i^\dag A_i = I$. Then we have
\beas
&& \E_l[\F(C\otimes I (\ket{\phi_{2,l}} \bra{\phi_{2,l}}), \ket{\phi_{3,l}} \bra{\phi_{3,l}}) ] \\
&=& \E_l \sum_{k_1,k_2\in [d]} |\beta_{k_1}|^2 |\beta_{k_2}|^2 \langle \varphi_{k_1l} | C\Big( \ket{\psi_{k_1l}} \bra{\psi_{k_2l}} \Big)  |\varphi_{k_2l} \rangle \\
&=& \E_l \sum_{i=1}^r \sum_{k_1,k_2\in [d]} |\beta_{k_1}|^2 |\beta_{k_2}|^2 \langle \varphi_{k_1l} | A_i  \ket{\psi_{k_1l}} \bra{\psi_{k_2l}} A_i^\dag  |\varphi_{k_2l} \rangle \\
&=& \sum_{i=1}^r  \E_l \left[\left|\sum_{k\in [d]} |\beta_{k}|^2 \langle \varphi_{kl} | A_i  \ket{\psi_{kl}} \right|^2 \right] \\
&\geq& \sum_{i=1}^r  \left| \E_l \left[\sum_{k\in [d]} |\beta_{k}|^2 \langle \varphi_{kl} | A_i  \ket{\psi_{kl}}  \right] \right|^2 .
\eeas

By Lemma \ref{lem1:general case}, the above equals
\beas
&=& \sum_{i=1}^r  \left| \sum_{k\in [d]} |\beta_{k}|^2 
\Big( (1-\delta) \E_l \bra{l} A_i  \ket{\hat{l}} + Err_i
\Big)\right|^2 \\
&=& \sum_{i=1}^r  \left| (1-\delta) \E_l \bra{l} A_i  \ket{\hat{l}} + Err_i \right|^2 .
\eeas
Expanding it, it is bounded from below by
\beas
&\geq& (1-\delta)^2 \sum_{i=1}^r |\E_l[\bra{l} A_i  \ket{\hat{l}}]|^2 
- 2 (1-\delta) \sum_{i=1}^r |\E_l[\bra{l} A_i  \ket{\hat{l}}]| |Err_i| \\
&\geq& (1-\delta)^2 \sum_{i=1}^r |\E_l[\bra{l} A_i  \ket{\hat{l}}]|^2 
- 2 (1-\delta) \sqrt{\sum_{i=1}^r |\E_l[\bra{l} A_i  \ket{\hat{l}}]|^2} \sqrt{\sum_{i=1}^r |Err_i|^2} \\
&\geq& (1-\delta)^2(1-\eta_1)
- 2 \sqrt{2\eta_1 |G|^2+18\delta}.
\eeas
In the above, we used the fact $1-\eta_1\leq \E_{k,l} \bra{k} C(\ket{\hat{k}} \bra{\hat{l}}) \ket{l}
=\sum_{i=1}^r |\E_l \bra{l} A_i \ket{\hat{l}} |^2 \leq 1$ and Lemma \ref{lem1:general case} again. In summary, we have 
\[
\E_l[\F(C\otimes I (\ket{\phi_{2,l}} \bra{\phi_{2,l}}), \ket{\phi_{3,l}} \bra{\phi_{3,l}}) ]
\geq (1-\delta)^2(1-\eta_1)
- 2 \sqrt{2\eta_1 |G|^2+18\delta} =: E_1.
\]

Similarly, we have
\[
\E_l[\F(P\otimes I (\ket{\phi_{4,l}} \bra{\phi_{4,l}}), \ket{\phi_{5,l}} \bra{\phi_{5,l}}) ] \geq 
(1-\delta)^2(1-\eta_2)
- 2\sqrt{2\eta_2 |G|^2+18\delta} =: E_2.
\]
The remaining part is similar to that of Theorem \ref{main theorem} based on Proposition \ref{prop:inequality for fidelity} and Jensen's inequality. Finally, we have
\beas
\E_l\Big[\F\Big( \rho_{7,l}, \ket{\phi_{7,l}} \bra{\phi_{7,l}} \Big)\Big] 
&\geq& 1 -  \sqrt{1-E_1} - \sqrt{1-E_2} .
\eeas
Substituting $\delta=2/(K-1)$ and $|G|=2K$ into the above lower bound, and simplifying it leads to the claimed result.
\end{proof}

From the above theorem, to make the bound small, we should set $\eta_1,\eta_2 \approx K^{-2}$.  

\section{Unitary channels}
\label{Section:discussions}

In this section, we consider two other situations that allow us to run the HHL algorithm in the worst case with a good average-case-correct QFT. We will mainly focus on unitary channels for the reasons explained below.

In the first case, we only use $C\in S_1(\eta)$ with an extra assumption that we have the ability to use $C^{-1}$. Here we need to assume that $C$ is an invertible quantum channel whose inverse is also a quantum channel, such as a unitary channel. A complete characterization of such quantum channels in the general case is given in \cite[Theorem 2.1]{nayak2007invertible}. In the special case that $C$ is a quantum channel from $L(\mathbb{C}^p)$ to $L(\mathbb{C}^p)$ whose inverse is also a quantum channel, then $C$ must be a unitary channel. Here $L(\mathbb{C}^p)$ is the space of all linear maps from $\mathbb{C}^p$ to $\mathbb{C}^p$. This is the case for us when studying the HHL algorithm.

In the second case, we assume that we only have $C \in S_1(\eta), P \in T_2(\eta)$ and they are unitary channels. To ensure that the HHL algorithm still works, we will make an extra assumption that the average of the trace of $C\circ P$ is close to 1.

\subsection{What if we can use $C$ and $C^{-1}$ for $C\in S_1(\eta)$?}

As discussed above, in this case, $C$ can only be a unitary channel. 
So the notation can be simplified. Note that
$\F(\ket{\psi}, \ket{\phi})=|\langle \psi|\phi\rangle|^2$
and
$\T(\ket{\psi}, \ket{\phi}) = \sqrt{1-|\langle \psi|\phi\rangle|^2}$, 
so in terms of the trace distance, we have
\beas
S_1(\eta) &=& \Big\{C \text{ is a unitary}: 
\E_k[\T(C\ket{\hat{k}}, \, \ket{k})^2 ] \leq \eta \Big\}, \\
S_2(\eta) &=& \Big\{C \text{ is a unitary}: 
\E_k[\T(C\ket{k}, \, \ket{\widehat{-k}})^2 ] \leq \eta \Big\}, \\
S_3(\eta) &=& \Big\{C \text{ is a unitary}: |\E_{k} \bra{k} C\ket{\hat{k}}|^2 \geq 1-\eta \Big\}.
\eeas
In the above, we used trace distance because it has the triangle inequality. This simplifies the analysis and, more importantly, makes the final results better.

As before, we first focus on the ideal case that all the eigenvalues are specified by $n$ bits exactly. Since now $C$ is unitary, there is no need to use the notation of density operators. Namely, we shall write $\ket{\psi_{7,l}}$ rather than $\rho_{7,l}$ as the output. The middle states will be denoted as $\ket{\psi_{1,l}}, \ldots, \ket{\psi_{6,l}}$. Recall that $\ket{\phi_{1,l}}, \ldots, \ket{\phi_{7,l}}$ are the states in the HHL algorithm when the QFT and its inverse have no errors.

\begin{thm}
\label{thm1:case with C inverse}
Using the same notation as Theorem \ref{main theorem}. Assume that $C\in S_1(\eta)$, and we can use $C^{-1}$ exactly. Furthermore, assume that the eigenvalues of $A$ are represented exactly by $n$ bits.
Then
\be
\E_l[\T(  \ket{\phi_{7,l}}, \ket{\psi_{7,l}} )^2]
\leq 
2\sqrt{2\eta}.
\ee
\end{thm}

\begin{proof}
We assume that $C \ket{\hat{k}} = \sum_l C_{k,l} \ket{l}$. In the HHL algorithm, if we replace $F_N^{-1}$ with $C$ and $F_N$ with $C^{-1}$, then at step 5, we obtain
$
\ket{\psi_{5,l}} = \ket{\psi_{5,l}'} +  \ket{G_l},
$
where 
\[
\ket{\psi_{5,l}'} 
= \sum_{k\in[d]} \beta_k |C_{p_k+l,p_k+l}|^2  \, \ket{p_k+l} \ket{\u_k}  \otimes \left(f(\sigma_k) \ket{0} + \sqrt{1-f(\sigma_k)^2} \ket{1}\right),
\]
and $\ket{G_l}$ is a garbage state that is orthogonal to $\ket{\psi_{5,l}'}$. So we have
\[
\E_l [|\langle \phi_{5,l}\ket{\psi_{5,l}} |]
\geq 
|\E_l [\langle \phi_{5,l}\ket{\psi_{5,l}'} ]|
-\E_l[\|\ket{G_l}\|].
\]
Regarding the first term, we know that
\[
\E_l [\langle \phi_{5,l}\ket{\psi_{5,l}'} ]
=\E_l \left[\sum_{k\in[d]} |\beta_k|^2 |C_{p_k+l,p_k+l}|^2 \right] 
\geq 1 - \eta.
\]
Regarding the second term, we know that
\[
\E_l[\|\ket{G_l}\|]^2 
\leq 
\E_l[\|\ket{G_l}\|^2]
= 
\E_l[1-\|\ket{\psi_{5,l}'} \|^2]
= 1 - \sum_{k\in[d]} |\beta_k|^2 \E_l[|C_{p_k+l,p_k+l}|^4]
\leq 1 - (1-\eta)^2.
\]
Therefore,
\[
\E_l [|\langle \phi_{5,l}\ket{\psi_{5,l}} |^2]
\geq
\E_l [|\langle \phi_{5,l}\ket{\psi_{5,l}} |]^2
\geq \left(1 - \eta - \sqrt{1 - (1-\eta)^2}\right)^2
\geq 1 - 2\sqrt{2\eta}.
\]
The last two steps of the HHL algorithm are assumed to be perfect with no error. So we obtain the claimed result.
\end{proof}

The general case can be analyzed similarly.

\begin{thm}
\label{thm2:case with C inverse}
Using the same notation as Theorem \ref{main theorem}. Assume that $C\in S_1(\eta)$, and we can use $C^{-1}$ exactly. 
Then for any integer $K\geq 1$, we have
\be
\E_l[\T(  \ket{\phi_{7,l}}, \ket{\psi_{7,l}} )^2]
\leq \frac{4\sqrt{2}}{\sqrt{K-1}} + 8 K \sqrt{\eta}.
\ee
\end{thm}

\begin{proof}
We shall use the notation in the proof of Theorem \ref{main theorem:general case}. We assume that $C \ket{\hat{k}} = \sum_l C_{k,l} \ket{l}$. When we apply $C\otimes I$ to $\ket{\phi_{2,l}}$, we obtain the state
\[
\ket{\psi_{3,l}}
= \sum_{k\in[d]}  \beta_k \sum_{g\in G_k} \alpha_{kg} \left(C_{g+l,g+l} \ket{g+l}
+\sum_{t:t\neq g+l} C_{g+l,t} \ket{t} \right) \ket{\u_k} + Err,
\]
where $|Err| \leq \sqrt{2/(K-1)}$ by \cite[Proposition 2]{linden2021average}. Next, we will apply control rotations. For convenience, for any $t\in[N]$ we denote the control rotation as $R(t)\ket{0}:=f(t/N) \ket{0} + \sqrt{1-f(t/N)^2} \ket{1}.$ Then we obtain 
\[
\ket{\psi_{4,l}}
= \sum_{k\in[d]}  \beta_k \sum_{g\in G_k} \alpha_{kg} \left(C_{g+l,g+l} \ket{g+l} \otimes R(g) \ket{0}
+\sum_{t:t\neq g+l} C_{g+l,t} \ket{t}  \otimes R(t-l) \ket{0} \right) \ket{\u_k} + Err.
\]
Finally, we apply $C^{-1}$ to the first register to obtain
\beas
\ket{\psi_{5,l}}
&=& \sum_{k\in[d]}  \beta_k \sum_{g\in G_k} \alpha_{kg} C_{g+l,g+l} \left(\bar{C}_{g+l,g+l}\ket{\widehat{g+l}}+ \sum_{t\neq g+l} \bar{C}_{t,g+l} \ket{\,\hat{t}\,} \right) \otimes R(g) \ket{0} \otimes \ket{\u_k} \\
&& +\, \sum_{k\in[d]}  \beta_k \sum_{g\in G_k} \alpha_{kg} \sum_{t:t\neq g+l} C_{g+l,t} C^{-1}\ket{t}  \otimes R(t-l) \ket{0} \otimes \ket{\u_k} + Err.
\eeas

In the following, we estimate the error between $\ket{\psi_{5,l}}$ and $\ket{\phi_{5,l}}$, the exact result obtained in the HHL algorithm. Recall that
\[
\ket{\phi_{5,l}} = \sum_{k\in[d]} \beta_k \sum_{g\in G_k} \alpha_{kg} \ket{\widehat{g+l}} \otimes R(g) \ket{0} \otimes \ket{\u_k} + Err
=:\ket{\phi_{5,l,1}} + Err.
\]
We denote $\ket{\psi_{5,l}}=\ket{\psi_{5,l,1}} + \ket{\psi_{5,l,2}} + \ket{\psi_{5,l,3}} + Err$ according to the three main summations. Note that $\sum_{g\in G_k}|\alpha_{kg}|^2\geq 1-2/(K-1)$, $|G_k|=2K$ and $C\in S_1(\eta)$, we have
\[
\E_l \langle \phi_{5,l,1}| \psi_{5,l,1} \rangle
= \E_l \left[\sum_{k\in[d]} |\beta_k|^2 \sum_{g\in G_k} |\alpha_{kg}|^2 |C_{g+l,g+l}|^2\right] \geq (1-\frac{2}{K-1}) (1-\eta).
\]
About the second term, we have
\beas
\E_l\langle \phi_{5,l,1}| \psi_{5,l,2} \rangle 
&=& \E_l \left[
\sum_{k\in[d]} |\beta_k|^2 \sum_{g\neq g'\in G_k} \alpha_{kg} \bar{\alpha}_{kg'} C_{g+l,g+l} \bar{C}_{g'+l,g+l} \bra{0} R(g')^\dag R(g) \ket{0}
\right] \\
&\leq&
\sum_{k\in[d]} |\beta_k|^2 \sum_{g\neq g'\in G_k} |\alpha_{kg} \bar{\alpha}_{kg'}| \, \E_l  [|\bar{C}_{g'+l,g+l}| ] \\
&\leq& 
\sum_{k\in[d]} |\beta_k|^2 \sum_{g\neq g'\in G_k} |\alpha_{kg} \bar{\alpha}_{kg'}| \, 
\sqrt{1 - \E_l  [|\bar{C}_{g+l,g+l}|^2 ] } \\
&\leq& 
2K \sqrt{\eta } .
\eeas
Finally,
\beas
\E_l\langle \phi_{5,l,1}| \psi_{5,l,3} \rangle &=& \E_l \left[\sum_{k\in[d]} |\beta_k|^2
\sum_{g,g'\in G_k} \sum_{t:t\neq g+l} \alpha_{kg} \alpha_{kg'} C_{g+l,t} \bar{C}_{g'+l,t}  \bra{0} R(g')^\dag R(t-l) \ket{0}\right] \\
&\leq& \sum_{k\in[d]} |\beta_k|^2
\sum_{g,g'\in G_k} |\alpha_{kg} \alpha_{kg'}| \, \E_l \left[ \sum_{t:t\neq g+l} |C_{g+l,t} \bar{C}_{g'+l,t}| \right] .
\eeas
In the above summation, if $g=g'$, then 
\[
\sum_{g\in G_k} |\alpha_{kg}|^2 \, \E_l \left[ \sum_{t:t\neq g+l} |C_{g+l,t}|^2 \right] =
\sum_{g\in G_k} |\alpha_{kg}|^2 \, \E_l [ 1 - |C_{g+l,g+l}|^2 ] \leq (1 - \frac{2}{K-1}) \eta.
\]
If $g\neq g'$, then
\beas
&& \sum_{g\neq g'\in G_k} |\alpha_{kg} \alpha_{kg'}| \, \E_l \left[ \sum_{t:t\neq g+l} |C_{g+l,t} \bar{C}_{g'+l,t}| \right] \\
&=& 
\sum_{g\neq g'\in G_k} |\alpha_{kg} \alpha_{kg'}| \, \E_l \left[ |C_{g+l,g'+l} \bar{C}_{g'+l,g'+l}| \right]
+ \sum_{g\neq g'\in G_k} |\alpha_{kg} \alpha_{kg'}| \, \E_l \left[ \sum_{t:t\neq g+l,g'+l} |C_{g+l,t} \bar{C}_{g'+l,t}| \right] \\
&\leq& 
\sum_{g\neq g'\in G_k} |\alpha_{kg} \alpha_{kg'}| \, \E_l \left[ |C_{g+l,g'+l}| \right]
+ \frac{1}{2}\sum_{g\neq g'\in G_k} |\alpha_{kg} \alpha_{kg'}| \, \E_l \left[ \sum_{t:t\neq g+l,g'+l} |C_{g+l,t}|^2+ |\bar{C}_{g'+l,t}|^2 \right] \\
&\leq& 
|G_k|\sqrt{\eta} + |G_k| \eta
=2K(\sqrt{\eta} + \eta).
\eeas
Hence
\[
\E_l\langle \phi_{5,l}| \psi_{5,l,3} \rangle \leq 2K(\sqrt{\eta} + \eta) + (1-\frac{2}{K-1}) \eta.
\]
Combining the above estimations and recalling that $|Err| \leq \sqrt{2/(K-1)}$, we have
\beas
\E_l[\T(  \ket{\phi_{7,l}}, \ket{\psi_{7,l}} )^2] 
&=& \E_l[\T(  \ket{\phi_{5,l}}, \ket{\psi_{5,l}} )^2] \\
&=& 1-\E_l[|\langle \phi_{5,l}| \psi_{5,l} \rangle|^2] \\
&\leq& 1- \left( (1-\frac{2}{K-1})(1-2\eta) -4K\sqrt{\eta}- 2K\eta  - \frac{2\sqrt{2}}{\sqrt{K-1}} - \frac{2}{K-1} \right)^2.
\eeas
Simplifying this leads to the claimed result.
\end{proof}

\subsection{What if we only have $S_1(\eta)$ and $T_2(\eta)$?}

Usually, $S_1(\eta)$ and $T_2(\eta)$ are not sufficient for us to make sure that the HHL algorithm still works. This can be seen easily by considering unitary channels and $\eta=0$. Now $C=\Phi F_N^{-1}, P=F_N \Psi$ for some diagonal unitaries $\Phi, \Psi$. Suppose that all eigenvalues are represented by $n$ bits; we have
$
\sum_{k\in[d]} \beta_k \ket{\widehat{p_k+l}} \ket{\u_k}
$
at the second step of the HHL algorithm (see Subsection \ref{subsection:Outline of HHL}).
If we apply $C$ to the first register, we obtain
$
\sum_{k\in[d]} \beta_k e^{i\Phi_k} \ket{p_k+l} \ket{\u_k}.
$
We then do control rotations and apply $P$ to the first register; the resulting state is
\[
\sum_{k\in[d]} \beta_k e^{i(\Phi_k+\Psi_k)} \ket{p_k+l} \ket{\u_k} \otimes \left(f(\sigma_k) \ket{0} + \sqrt{1-f(\sigma_k)^2} \ket{1}\right).
\]
After the last two steps of the HHL algorithm listed in Subsection \ref{subsection:Outline of HHL}, we finally obtain
\[
\sum_{k\in[d]} \beta_k e^{i(\Phi_k+\Psi_k)} f(\sigma_k)  \ket{\u_k} \otimes  \ket{0} + \ket{0}^\bot.
\]
As we can see, this does not include $f(A)\ket{\b} = \sum_{k\in[d]} \beta_k f(\sigma_k)  \ket{\u_k}$ exactly  because of the unknown phases $e^{i(\Phi_k+\Psi_k)}$. 
One natural extra assumption we can make to fix this is
$|\E_k[\bra{k} C\circ P \ket{k}]| \geq 1-\eta$. If this happens, then $\Phi = \Psi^{-1}$ in the above situation when $\eta=0$. The following result confirms this for unitary channels.

\begin{lem}
\label{lem:case with C, P and more}
Assume that $C\in S_1(\eta_1), P \in T_2(\eta_2)$ are unitary channels. If
$|\E_k[\bra{k} C\circ P \ket{k}]| \geq 1-\eta_3,$ then
\[
|\E_k [\bra{\hat{k}} P \ket{k} \bra{k} C \ket{\hat{k}}] | \geq 1 - (\eta_1+\eta_2+\eta_3).
\]
\end{lem}

\begin{proof}
Assume that $P \ket{k} = \sum_l P_{kl} \ket{\hat{l}} $ and $ C\ket{\hat{k}} = \sum_l C_{kl} \ket{l}$. Then
\[
\E_k[\bra{k} C\circ P \ket{k}]
=\E_k [P_{kk} C_{kk}] + \E_k \left[\sum_{l:l\neq k} P_{kl} C_{lk}\right].
\]
Regarding the second term, since $C\in S_1(\eta_1), P \in T_2(\eta_2)$, we have
\beas
 \E_k \left[\sum_{l:l\neq k} P_{kl} C_{lk}\right] \leq 
 \frac{1}{2} \E_k \left[\sum_{l:l\neq k} |P_{kl}|^2 + \sum_{l:l\neq k} |C_{lk}|^2\right] 
 = 
 \frac{1}{2} \E_k \left[ 2 - |P_{kk}|^2 - |C_{kk}|^2\right] 
 \leq \eta_1+\eta_2 .
\eeas
Thus
\[
1-\eta_3 \leq |\E_k[\bra{k} C\circ P \ket{k}]|
\leq |\E_k [P_{kk} C_{kk}]| + \eta_1+\eta_2.
\]
This leads to the claimed result.
\end{proof}


With the above preliminaries, in the general case, the extra assumption we can make is $\E_{k,l}\bra{k} \, C\circ P (\ket{k} \bra{l}) \, \ket{l} \geq 1-\eta_3.$ Note that this would imply $C\in S_3(\eta)$ when $P=F_N$. For quantum channels, it may not be easy to find a lightweight protocol to verify this assumption.
Because of this, we now assume that $C, P$ are unitary channels. In this case, the assumption $|\E_k[\bra{k} C\circ P \ket{k}]| \geq 1-\eta_3$ made in Lemma \ref{lem:case with C, P and more} is easy to verify. The algorithm reads as follows: We randomly choose $k$, then run $C\circ P$ on it and measure in the computational basis. Output 1 if the outcome is $k$, and output 0 otherwise.

Using a similar argument to the proof of Theorems \ref{thm1:case with C inverse} and \ref{thm2:case with C inverse}, and combining Lemma \ref{lem:case with C, P and more}, we obtain the following results.

\begin{thm}
\label{thm1:case with C, P and more}
Using the same notation as Theorem \ref{main theorem}. Assume that $C\in S_1(\eta_1), P \in T_2(\eta_2)$ and $|\E_k[\bra{k} C\circ P \ket{k}]| \geq 1-\eta_3$. Further, assume that the eigenvalues of $A$ are represented by $n$ bits precisely.
Then
\be
\E_l[\T(  \ket{\phi_{7,l}}, \ket{\psi_{7,l}} )^2]
\leq 
2\sqrt{2(\eta_1+\eta_2+\eta_3)}.
\ee
\end{thm}

\begin{thm}
\label{thm2:case with C, P and more}
Using the same notation as Theorem \ref{main theorem}. Assume that $C\in S_1(\eta_1), P \in T_2(\eta_2)$ and $|\E_k[\bra{k} C\circ P \ket{k}]| \geq 1-\eta_3$.
Then for any integer $K\geq 1$, we have
\be
\E_l[\T(  \ket{\phi_{7,l}}, \ket{\psi_{7,l}} )^2]
\leq 
\frac{4\sqrt{2}}{\sqrt{K-1}} + 8 K \sqrt{\eta_1+\eta_2+\eta_3}.
\ee
\end{thm}

\section{Conclusions}

In this work, we identified a strengthened average-case correctness condition for the QFT that is sufficient for running the HHL algorithm in the worst case. Using a strengthened Linden-de Wolf’s protocol, we can also run the HHL algorithm for solving linear systems. As mentioned in \cite{linden2021average}, it would be interesting to explore more applications of this idea. 
Our results suggest that similar coherence-sensitive average-case conditions may be useful for analyzing other QFT-based quantum algorithms.
For example, a similar idea to Linden-de Wolf has been used in \cite{van2023quantum} for the task of quantum tomography.

\section*{Acknowledgements}

The research was supported by the National Key Research Project of China under Grant No. 2025YFA1017200.  The author would like to thank Noah Linden for helpful discussions on an earlier version of this paper.

\bibliographystyle{plain}
\bibliography{main}

\end{document}